\documentclass[a4paper,UKenglish,thm-restate]{lipics-v2019}

\title{Dynamic complexity of Reachability: \texorpdfstring{\\}{}How many changes can we handle?}
\titlerunning{Dynamic complexity of Reachability: How many changes can we handle?}
\author{Samir Datta}{Chennai Mathematical Institute, India}{sdatta@cmi.ac.in}{}{The author was partially funded by a grant from Infosys foundation and SERB-MATRICS grant MTR/2017/000480.}
\author{Pankaj Kumar}{Chennai Mathematical Institute, India \and Department of Applied Mathematics, Charles University in Prague, Czechia}{pankaj@kam.mff.cuni.cz}{}{The author was supported by Czech Science Foundation GA{\v C}R (grant \#19-27871X).}
\author{Anish Mukherjee}{Institute of Informatics, University of Warsaw, Poland}{anish@mimuw.edu.pl}{https://orcid.org/0000-0002-5857-9778}{The author was supported by the ERC CoG grant TUgbOAT no 772346.}
\author{Anuj Tawari}{Chennai Mathematical Institute, India}{atawari@cmi.ac.in}{}{}
\author{Nils Vortmeier}{TU Dortmund, Germany}{nils.vortmeier@tu-dortmund.de}{}{The author acknowledges the financial support by DFG grant SCHW 678/6-2.}
\author{Thomas Zeume}{Ruhr-Universität Bochum, Germany}{thomas.zeume@rub.de}{}{}

\authorrunning{S. Datta, P. Kumar, A. Mukherjee, A. Tawari, N. Vortmeier, T. Zeume}%

\Copyright{Samir Datta, Pankaj Kumar, Anish Mukherjee, Anuj Tawari, Nils Vortmeier and Thomas Zeume}%

\ccsdesc[500]{Theory of computation~Complexity theory and logic}
\ccsdesc[500]{Theory of computation~Finite Model Theory}

\keywords{Dynamic complexity, reachability, complex changes}%

\relatedversion{This is the full version of a paper to be presented at ICALP 2020.}

\acknowledgements{The first author would like to thank Chetan Gupta for finding a problem in a previous version of the proof of Theorem~\ref{theorem:btw}.}%

\nolinenumbers %

\hideLIPIcs  %

\EventEditors{Artur Czumaj, Anuj Dawar, and Emanuela Merelli}
\EventNoEds{3}
\EventLongTitle{47th International Colloquium on Automata, Languages, and Programming (ICALP 2020)}
\EventShortTitle{ICALP 2020}
\EventAcronym{ICALP}
\EventYear{2020}
\EventDate{July 8--11, 2020}
\EventLocation{Saarbrücken, Germany (virtual conference)}
\EventLogo{}
\SeriesVolume{168}
\ArticleNo{122}

\usepackage{mathtools}

\usepackage{stackengine}
\stackMath
\usepackage{graphicx}
\usepackage{xspace}
\usepackage{soul}
\usepackage{tikz}

\newcommand{\Reach}{\myproblem{Reach}}

\newcommand{\RelName}[1]{\mtext{#1}}
\newcommand{\tpl}{\bar}

\newcommand{\mtext}[1]{\textsc{#1}}

\newcommand{\schema}{\ensuremath{\sigma}\xspace}
\newcommand{\adom}{\ensuremath{\text{adom}}}

\newcommand{\query}{\ensuremath{Q}}
\newcommand{\db}{\ensuremath{\calI}\xspace}%
\newcommand{\inp}{\ensuremath{\calI}\xspace}
\newcommand{\aux}{\ensuremath{\calA}\xspace}%

\newcommand{\struc}{\calS}
\newcommand{\ans}{\RelName{Ans}}

\newcommand{\restrict}[2]{#1[#2]}
\newcommand{\N}{\ensuremath{\mathbb{N}}}
\newcommand{\Z}{\ensuremath{\mathbb{Z}}}
\newcommand{\R}{\ensuremath{\mathbb{R}}}

\newcommand{\bigO}{\ensuremath{\mathcal{O}}}

\newcommand{\df}{\ensuremath{\mathrel{\smash{\stackrel{\scriptscriptstyle{
    \text{def}}}{=}}}}}

\newcommand  {\myclass} [1]  {\ensuremath{\textsf{\upshape #1}}}

\newcommand{\StaClass}[1]{\myclass{#1}\xspace}

\newcommand{\DynClass}[1]{\myclass{Dyn#1}\xspace}

\newcommand  {\myproblem} [1] {\textsc{#1}\xspace}

\newcommand     {\LOGSPACE}     {\StaClass{LOGSPACE}}
\newcommand     {\NL}   {\StaClass{NL}}

\newcommand     {\AC}   {\StaClass{AC}}
\newcommand{\ACz}{\mbox{\myclass{AC}$^0$}\xspace}

\newcommand{\FO}{\StaClass{FO}}

\newcommand{\FOParityar}{\StaClass{FO+Mod$\:2(\leq,+,\times)$}}
\newcommand{\FOMajar}{\StaClass{FO+Maj$(\leq,+,\times)$}}
\newcommand{\FOar}{\StaClass{FO$(\leq,+,\times)$}}%

\newcommand{\CQ}[1][]{\StaClass{CQ}}
\newcommand{\UCQ}[1][]{\StaClass{UCQ}}
\newcommand{\CQneg}[1][]{\StaClass{CQ\ensuremath{^{\mneg}}}}
\newcommand{\UCQneg}[1][]{\StaClass{UCQ\ensuremath{^{\mneg}}}}

\newcommand{\mneg}{\neg} %

\newcommand{\DynFO}{\DynClass{FO}}
\newcommand{\DynAC}{\DynClass{AC}}
\newcommand{\DynFOar}{\DynClass{FO$(\leq,+,\times)$}}
\newcommand{\DynFOMaj}{\DynClass{FO+Maj}}
\newcommand{\DynFOParity}{\DynClass{FO+Mod$\:2$}}
\newcommand{\FOMod}[1][2]{\StaClass{FO+Mod$\:$#1}}
\newcommand{\DynFOMod}[1][2]{\DynClass{FO+Mod$\:$#1}}
\newcommand{\DynFOParityar}{\DynClass{FO+Mod$\:2(\leq,+,\times)$}}
\newcommand{\DynFOMajar}{\DynClass{FO+Maj$(\leq,+,\times)$}}
\theoremstyle{plain}
\newtheorem{maintheorem}{Main Theorem}
\newtheorem*{question*}{Question}
\newtheorem*{openquestion*}{Open question}

\newenvironment{proofsketch}{\begin{proof}[Proof sketch.]}{\end{proof}}

\newenvironment{proofof}[1]{\begin{proof}[Proof (of #1).]}{\end{proof}}
\newenvironment{proofsketchof}[1]{\begin{proof}[Proof sketch (of #1).]}{\end{proof}}
\newenvironment{proofenum}{\begin{enumerate}[(a)]}{\qedhere \end{enumerate}}%

\providecommand {\calA}      {{\mathcal A}\xspace}
\providecommand {\calB}      {{\mathcal B}\xspace}
\providecommand {\calC}      {{\mathcal C}\xspace}

\providecommand {\calG}      {{\mathcal G}\xspace}

\providecommand {\calI}      {{\mathcal I}\xspace}

\providecommand {\calP}      {{\mathcal P}\xspace}

\providecommand {\calS}      {{\mathcal S}\xspace}
\providecommand {\calT}      {{\mathcal T}\xspace}

\providecommand {\calW}      {{\mathcal W}\xspace}

\newcommand{\prog}{\ensuremath{\calP}\xspace}

\newcommand{\state}{\ensuremath{\struc}\xspace}

\newcommand{\auxSchema}{\ensuremath{\schema_{\text{aux}}}\xspace}

\newcommand{\domain}{\ensuremath{ D}\xspace}

\begin{document}

\maketitle

\begin{abstract}

In 2015, it was shown that reachability for arbitrary directed graphs can be updated by first-order formulas after inserting or deleting single edges. Later, in 2018, this was extended for changes of size $\frac{\log n}{\log \log n}$, where $n$ is the size of the graph. Changes of polylogarithmic size can be handled when also majority quantifiers may be used.

In this paper we extend these results by showing that, for changes of polylogarithmic size, first-order update formulas suffice for maintaining (1) undirected reachability, and (2) directed reachability under insertions. For classes of directed graphs for which efficient parallel algorithms can compute non-zero circulation weights, reachability can be maintained with update formulas that may use ``modulo 2'' quantifiers under changes of polylogarithmic size. 
Examples for these classes include the class of planar graphs and graphs with bounded treewidth. The latter is shown here.

As the logics we consider cannot maintain reachability under changes of larger sizes, our results are optimal with respect to the size of the changes.
\end{abstract}

\section{Introduction}

Suppose we are given a graph $G$ whose edge relation is subjected to insertions and deletions of edges. Which resources are required to update the reachability relation of the graph?

Recently it was shown that if one is allowed to store auxiliary relations, then the reachability relation can be updated after single edge insertions and deletions using first-order logic formulas with access to the graph, the stored relations, and the changed edges \cite{DattaKMSZ18}. In other words, the reachability query is contained in the dynamic complexity class $\DynFO$ \cite{PatnaikI97}. From a database perspective, this means that it can be updated with core-SQL queries;  from the perspective of circuit complexity, this means that reachability can be updated by circuits of polynomial size in constant-time due to the correspondence of first-order logic and $\AC^0$ established by Barrington, Immerman, and Straubing \cite{BarringtonIS90}. 

Understanding single edge insertions and deletions is an important first step. Yet in applications, changes to a graph $G$ often come as bulk set $\Delta E$ of changed edges. It is natural to ask, how large the set $\Delta E$ of edges can be such that reachability can be maintained with the same resources as for single edge changes  -- that is with first-order formulas or, respectively, $\AC^0$ circuits. Using existing lower bounds for circuits \cite{Smolensky1987}, it is easy to see that $\DynFO$ (or $\DynAC^0$, respectively) cannot handle changes of size larger than polylogarithmic for many queries, including the reachability query (see Section \ref{section:barriers} for a more detailed discussion).

The best one can hope for is to maintain reachability with first-order formulas for changes of polylogarithmic size with respect to the size of the graph. In a first step, a subset of the authors showed that reachability can be maintained in $\DynFOar$ under changes of size $O(\frac{\log n}{\log \log n})$ \cite{DattaMVZ18}. Here, the class $\DynFOar$ extends $\DynFO$ by access to built-in arithmetic, which for technical reasons is more natural for bulk changes. Unfortunately, the techniques used in \cite{DattaMVZ18} seem only be able to handle changes of polylogarithmic size in the extension of $\DynFO$ by majority quantifiers, that is, in the class \DynFOMajar.

In this paper we make progress on handling changes of polylogarithmic size in $\DynFO$ by attacking the challenge from two directions.  First, we establish two restrictions for which reachability can be maintained under these changes. 

\begin{maintheorem}
  Reachability can be maintained in $\DynFOar$  under
  \begin{itemize}
   \item insertions of polylogarithmically many edges; and
   \item insertions and deletions of polylogarithmically many edges if the graph remains undirected.
  \end{itemize}
\end{maintheorem}

As second contribution of this paper, we provide a meta-theorem for establishing classes of graphs for which reachability can be maintained under polylogarithmic-size changes with a slight extension of first-order logic. In this extension, $\DynFOParityar$, formulas used for updating the reachability information and the auxiliary relations may use parity quantifiers in addition to the traditional universal and existential quantifiers. 

\begin{maintheorem}
  Reachability can be maintained in $\DynFOParityar$ under insertions and deletions of polylogarithmically many edges on classes of graphs for which polynomially bounded non-zero circulation weights can be computed in $\AC$.
\end{maintheorem}

Here a weighting function for the edges of a graph has \emph{non-zero circulation}, if the weight of every directed cycle is non-zero (see Section~\ref{section:isolation} for details). The class $\AC$ contains queries computable by circuits of polynomial size and polylogarithmic depth. 
Examples for graph classes for which non-zero circulation weights can be computed in \AC include the class of planar graphs and graphs with bounded treewidth. The latter is shown here.
We note that isolating weights, a concept closely related to non-zero circulation weights, have been used previously in dynamic complexity for establishing that reachability is in non-uniform $\DynFOParityar$ under single edge changes \cite{DattaHK14}, a precursor result to reachability~in~$\DynFO$. 

For our results, we employ two techniques of independent interest. The first technique relies on the power of first-order logic on structures of polylogarithmic size. It is well-known that  reachability can be computed by a uniform circuit family of size $N^{\bigO(N^{1/d})}$ and depth $2d$. An immediate consequence is that all $\NL$-queries can be expressed by first-order formulas for graphs with $n$ nodes but only polylogarithmically many edges.  Thus, for maintaining a query under changes $\Delta E$ of polylogarithmic size, a dynamic program can (1) do an arbitrary $\NL$-computation on $\Delta E$, and (2) update the auxiliary data by combining the computed information with the previous auxiliary data using a first-order formula.

The second technique we rely on is a slight generalization of the ``Muddling Lemma'' from~\cite{DattaMSVZ19}. The Muddling Lemma reduces the requirements for proving that a query is in $\DynFO$: a query is in $\DynFO$ if, essentially, one can update the query for polylog many steps starting from auxiliary data precomputed in $\AC$. Here we observe that this can be strengthened for changes of polylogarithmic size: a query is in $\DynFO$ if, essentially,  one can update the query under one polylogarithmic-size change from auxiliary data precomputed in~$\AC$.

Parts of the results presented here have been included in the PhD thesis of Nils Vortmeier~\cite{Vortmeier19}.

\subparagraph*{Outline} After recalling the dynamic complexity framework in Section \ref{section:framework}, we shortly outline barriers for the size of bulk changes in Section \ref{section:barriers} and recall useful techniques in Section \ref{section:tools}. Afterwards we present our results for $\DynFO$ in Section \ref{section:dynfo} and for $\DynFOParityar$ in Section \ref{section:isolation}.

\section{The dynamic setting}
\label{section:framework}
We briefly repeat the essentials of dynamic complexity, closely following \cite{DattaMVZ18} which in turn builds on \cite{SchwentickZ16}. The goal of a \emph{dynamic program} is to answer a given query on a relational \emph{input structure} subjected to changes that insert tuples into the input relations or delete tuples from them. The program may use  auxiliary information represented by an \emph{auxiliary structure} over the same domain as the input structure. Initially, both input and auxiliary structure are empty; and the domain is fixed during each run of the program. Whenever a change to the input structure occurs, the auxiliary structure is updated by means of first-order formulas.

\subparagraph*{Changes}

For a (relational) structure $\db$ over domain $\domain$ and schema $\schema$, a change $\Delta \db$ consists of sets $R^{+}$ and $R^{-}$ of tuples for each relation symbol $R \in \schema$. The result $\db + \Delta \db$ of an application of the change $\Delta \db$  to $\db$ is the input structure where $R^{\db}$ is changed to $(R^{\db} \cup R^{+}) \setminus R^{-}$. The \emph{size} of $\Delta \db$ is the total number of tuples in relations $R^{+}$ and $R^{-}$ and the set of \emph{affected elements} is the (active) domain of tuples in $\Delta \db$.

\subparagraph*{Dynamic Programs and Maintenance of Queries} A dynamic program consists of a set of update rules that specify how auxiliary relations are updated after changing the input structure. Let $\db$ be the current input structure over schema $\schema$ and let $\aux$ be the auxiliary structure over some schema $\auxSchema$.
An \emph{update rule} for updating an $\ell$-ary auxiliary relation $T$ after a change is a first-order formula $\varphi$ over schema $\schema \cup \auxSchema$ with $\ell$ free variables. After a change $\Delta \db$, the new version of $T$ is $T \df \{ \tpl a \mid (\db + \Delta \db, \aux) \models \varphi(\tpl a)\}$, so, the updated auxiliary relation includes all tuples $\tpl a$ such that $\varphi(\tpl a)$ is satisfied when it is evaluated on the changed input structure and the old auxiliary structure. 
Note that a dynamic program can choose to have access also to the old input structure by storing it in its auxiliary relations. %

For a state $\state = (\db, \aux)$ of the dynamic program $\prog$ with input structure $\db$ and auxiliary structure $\aux$ we denote the state of the program after applying a change sequence $\alpha$ and updating the auxiliary relations accordingly by $\prog_\alpha(\state)$. 

The dynamic program \emph{maintains} a $q$-ary query $\query$ under changes of size $k$ if it has a $q$-ary auxiliary relation $\ans$ that at any time stores the result of $\query$ applied to the current input structure. More precisely, for each non-empty sequence $\alpha$ of changes of size $k$, the relation $\ans$ in $\prog_\alpha(\state_\emptyset)$ and $\query(\alpha(\db_\emptyset))$ coincide, where the state $\state_\emptyset \df (\db_\emptyset, \aux_\emptyset)$ consists of an input structure $\db_\emptyset$ and an auxiliary structure $\aux_\emptyset$ over some common domain that both have empty relations, and $\alpha(\db_\emptyset)$ is the input structure after applying $\alpha$. 

If a dynamic program maintains a query, we say that the query is in $\DynFO$. Similarly to \DynFO one can define variants with built-in auxiliary relations and with more powerful update formulas. For instance, the class \DynFOar contains queries that can be maintained by first-order update formulas with access to three particular auxiliary relations $<, +, $ and $\times$ which are initialized as a linear order and the corresponding addition and multiplication relations;  in the class $\DynFOMod[p]$, update formulas may use modulo-$p$-quantifiers in addition to existential and universal quantifiers.

We state our results for dynamic classes with access to the arithmetic relations $\leq, +$ and $\times$. Handling bulk changes without access to arithmetic leads to technical issues which distract from the fundamental dynamic properties. See \cite{DattaMVZ18, Vortmeier19} for further discussions on this topic and how our results can be stated for \DynFO in an adapted setting which takes these technical issues into account. 

For the construction of dynamic programs in this paper we assume that changes either only insert edges or only delete edges. This is no restriction, as corresponding update formulas can be combined to process a change that inserts and deletes edges at the same time, by first processing the inserted edges and then processing the deleted edges.

\section{Barriers for the size of bulk changes}
\label{section:barriers}
In the following we outline why it is not possible to maintain reachability under changes of larger than polylogarithmic size with first-order formulas, even in the presence of parity quantifiers. 

The idea is simple. A classical result by Smolensky states that for computing the number of ones modulo a prime $q$ occurring in a bit string of length $n$, an $\AC[p]$ circuit of depth $d$ requires $2^{\Omega(n^{1/2d})}$ many gates, for each prime $p$ distinct from $q$ (see \cite{Smolensky1987} or, for a modern exposition, \cite[Theorem 12.27]{Jukna2012}). A simple, well-known reduction yields that deciding reachability for graphs with $n$ edges which are disjoint unions of paths  %
also requires $\AC[p]$ circuits of size~$2^{\Omega(n^{1/2d})}$. Indeed, computing the number of ones in $w = a_1 \cdots a_n$ modulo $q$ can be reduced  to reachability as follows. Consider the graph with nodes $\{(i, k) \mid 1 \leq i \leq n+1 \text{ and }  0 \leq k < q\}$ and edges $\{((i, k), (i+1, k)) \mid a_i = 0\} \cup \{((i, k), (i+1, k+1 \bmod q)) \mid a_i = 1\}$. It is easy to see that (i) the graph has $\bigO(n)$ edges and is a disjoint union of $q$ paths, 
and (ii) there is a path from $(1, 0)$ to $(n+1, 0)$ if and only if the number of ones in $w$ is $0$ modulo $q$.  

These lower bounds for circuit sizes immediately translate into lower bounds for first-order formulas with modulo $p$ quantifiers via the correspondence due to Barrington, 
Immerman, and Straubing \cite{BarringtonIS90}.

\begin{theorem}
  Let $f(n) \in \log^{\omega(1)} n$ be a function from $\N$ to $\N$ and let $p$ be a prime. There is no \FOMod[p] formula with access to built-in relations that defines 
  \begin{enumerate}
   \item whether the size of a unary relation $U$ with $|U| \leq  f(n)$ is divisible by $q$, for primes $q$ distinct from $p$;  
   \item reachability in graphs with at most $f(n)$ edges, even for disjoint unions of paths. 
  \end{enumerate}
\end{theorem}
\begin{proofsketch}
\begin{proofenum}
\item
Let $f(n)$ be some function from $\log^{\omega(1)} n$. Suppose, towards a contradiction, that there is an \FOMod[p] formula with access to built-in relations that defines whether the size of a unary relation $U$ with $|U| \leq  f(n)$ is divisible by $q$, for some primes $p \neq q$. 
Then, by \cite{BarringtonIS90}, for every $n$ there is an $\AC[p]$ circuit of some fixed depth $d$ that decides that question for inputs if size $n$, and the size of this circuit is polynomial in $n$.
That is a contradiction, as by Smolensky's lower bound every such circuit needs to have size $2^{\Omega(f(n)^{1/2d})} = 2^{\log(n)^{\omega(1)}} = n^{\omega(1)}$.
\item
This part can be proven analogously to Part (a), using the circuit lower bound for graph reachability.
\end{proofenum}
\end{proofsketch}

Those lower bounds have the immediate consequence that $\DynFO$ cannot deal with bulk changes of larger than polylogarithmic size. Indeed,  from any formula that updates the result of a query after an insertion of $f(n)$ tuples into an initially empty input relation one can construct a formula that defines the query for inputs of size $f(n)$.

\begin{corollary}\label{cor:barrier}
  Let $f(n) \in \log^{\omega(1)} n$ be a function from $\N$ to $\N$ and let $p$ be a prime. Then the following queries cannot be maintained  in $\DynFOMod[p]$  for bulk changes of size $\leq f(n)$, even if the auxiliary relations may be initialized arbitrarily:
  \begin{enumerate}
   \item divisibility of the size of a unary relation by a prime $q \neq p$, and
   \item reachability in graphs, even if restricted to disjoint unions of paths.
  \end{enumerate}
\end{corollary}

\section{Techniques and Tools}
\label{section:tools}
In the previous section we recalled that $\FOar$, even if equipped with modulo $p$ quantifiers, is not very expressive in general. That changes when we are only interested in small substructures of our input: $\FOar$ can express every $\NL$-computable query on subgraphs of polylogarithmic size.

\begin{theorem}\label{theorem:nl-fo}
Let $k$ and $c$ be arbitrary natural numbers, and let $\query$ be a $k$-ary, $\NL$-computable graph query. There is an \FOar formula $\varphi$ over schema $\{E, D\}$ such that for any graph $G$ with $n$ nodes, any subset $D$ of its nodes of size at most $\log^c n$ and any $k$-tuple $\tpl a \in {D}^k$: $\tpl a \in \query(\restrict{G}{D})$ if and only if $(G, D) \models \varphi(\tpl a)$. Here, $\restrict{G}{D}$ denotes the subgraph of $G$ induced by~$D$.
\end{theorem}
\begin{proof}
We prove the result for the reachability query. As reachability is $\NL$-complete under $\FOar$-reductions \cite{ImmermanDC}, and every $\FOar$-reduction maps an instance of size $\log^c n$ to an instance of size $\log^{cd} n$ for a fixed $d \in \N$, the full result follows.

It is well-known (see for example \cite[p. 613]{ChenOST16}), that for every $d \in \N$ there is a uniform circuit family for reachability where the circuit for inputs of size $N$ has depth $2d$ and size $N^{\bigO(N^{1/d})}$.
Suppose the input size $N$ is only $\log^c n$, for some $c \in \N$ and pick $d \df 2c$.
Then the circuit size
\begin{align*}
 N^{\bigO(N^{1/d})} = & (\log^c n)^{\bigO((\log^c n)^{1/d})} \\
  = & (\log n)^{c \bigO((\log n)^{c/d})} = (\log n)^{\bigO((\log n)^{c/2c})} = (\log n)^{\bigO(\sqrt{\log n})}\\
  = & 2^{\bigO(\log \log n \sqrt{\log n})} \subseteq 2^{\bigO(\sqrt{\log n} \sqrt{\log n})} = 2^{\bigO(\log n)} = n^{\bigO(1)}
\end{align*}
is polynomial in $n$, so, the circuit is a uniform $\ACz$ circuit for reachability for graphs of size $\log^c n$.
The existence of $\varphi$ follows by the equivalence of uniform \ACz and $\FOar$ \cite{BarringtonIS90}.
\end{proof}

The Muddling Lemma simplifies the maintenance of queries under single edge changes~\cite{DattaMSVZ19}. It states that for many natural queries $\query$, in order to show that $\query$ can be maintained, it is enough to show that the query can be maintained for a bounded number of steps. In the following we recall the necessary notions and extend the lemma to bulk changes. 

A query $\query$ is \emph{almost domain-independent} if there is a $c \in \N$ such that $\restrict{\query(\calA)}{(\adom(\calA) \cup B)} = \query(\restrict{\calA}{(\adom(\calA) \cup B)})$ for all structures $\calA$ and sets $B \subseteq A \setminus \adom(\calA)$ with $|B| \geq c$. Here, $\adom(\calA)$ denotes the \emph{active domain}, i.e.\ the set of domain elements that are used in some tuple of $\calA$.
A query $\query$ is \emph{$(\calC,f)$-maintainable}, for some complexity class $\calC$ and some function~\mbox{$f:\N\to\R$}, if there is a dynamic program $\prog$ and a $\calC$-algorithm $\mathbb A$ such that for each input structure $\db$ over a domain of size $n$, each linear order $\leq$ on the domain, and each change sequence $\alpha$ of length $|\alpha| \leq f(n)$, the relation $Q$ in $\prog_\alpha(\state)$ and $\query(\alpha(\db))$ coincide, where $\state = (\inp, \mathbb A(\inp,\leq))$.

The Muddling Lemma from \cite{DattaMSVZ19} has been formulated for bulk changes in \cite{DattaMVZ18,Vortmeier19}\footnote{The statement and proof in \cite{DattaMVZ18} is slightly flawed and has been corrected in \cite{Vortmeier19}.}. %
\begin{theorem}[\cite{DattaMVZ18, Vortmeier19}]\label{theorem:fewerChanges}
Let $\query$ be an \NL-computable, almost domain independent query, and let $c \in \N$ be arbitrary.
If the query $\query$ is $(\AC^d,\log^d n)$-maintainable under changes of size $\log^c n$ for some $d \in \N$, then $\query$ is in $\DynFOar$ under changes of size $\log^c n$.
\end{theorem}  

The previous theorem can be strengthened as follows.
\begin{restatable}{theorem}{theoMuddling}
\label{theorem:oneshotmuddling}
Let $\query$ be an \NL-computable, almost domain independent query, and let $c \in \N$ be arbitrary.
If  the query $\query$ is $(\AC^d,1)$-maintainable under changes of size $\log^{c+d} n$ for some $d \in \N$, then $\query$ is in $\DynFOar$ under changes of size $\log^c n$.
\end{restatable}
\begin{proof}
Let $\query$ and $d$ be as in the theorem statement, and let $\mathbb A$ be an $\AC^d$ algorithm and $\prog$ a dynamic program that witness that $\query$ is $(\AC^d,1)$-maintainable under changes of size $\log^{c+d} n$.
By Theorem~\ref{theorem:fewerChanges} it suffices to show that there is an $\AC^d$ algorithm $\mathbb A'$ and a dynamic program $\prog'$ that witness that $\query$ is $(\AC^d,\log^d n)$-maintainable under changes of size $\log^{c} n$.

We choose $\mathbb A'$ as $\mathbb A$. The program $\prog'$ just stores the at most $\log^{c+d} n$ changes that accumulate during the $\log^d n$ steps, and in each step uses $\prog$ to answer $\query$, using the initial auxiliary relations computed by $\mathbb A$. 
\end{proof}

\section{Handling Polylog Changes with \DynFO}
\label{section:dynfo}
\newcommand{\Vaff}{\ensuremath{V_{\text{aff}}}}

So far we do not know how to maintain directed reachability under polylogarithmically many changes in \DynFOar. 
In this section we show that reachability can be maintained in $\DynFOar$ under insertions of polylogarithmically many edges for arbitrary graphs (disallowing any deletions), and under insertions \emph{and} deletions of polylogarithmic size for undirected graphs. 

The general idea is similar in both cases. After changing polylogarithmically many edges with an effect on nodes $\Vaff$, the dynamic program (1) computes a structure of polylogarithmic size on $\Vaff$, (2) uses Theorem \ref{theorem:nl-fo} to compute helpful information for this structure, and (3) updates the auxiliary relations by combining this information with the previous auxiliary data. Both (1) and (3) are performed by first-order formulas, and (2) uses an $\NL$-computation.

\begin{theorem}\label{theorem:bulk:mso:reach}
Reachability is in $\DynFOar$ under insertions of size $\log^c n$, for every $c \in \N$.
\end{theorem}
\begin{proof}
Let $c \in \N$ be fixed. We construct a dynamic program with a single auxiliary relation $\ans$ which stores the transitive closure of the current graph.

Whenever a set $E^+$ of edges is inserted into the current graph $G = (V,E)$, the dynamic program updates $\ans$ with the help of the transitive closure relation of a graph $H$ defined as follows. The nodes of $H$ are the nodes $\Vaff$ affected by the change, that is, the nodes incident to edges in $E^+$. The edge set $E_H$ of $H$ contains the newly inserted edges $E^+$, and additionally edges $(u,v)$ for all pairs $(u,v)$ of nodes from $\Vaff$ that are connected by a path in $G$. Observe that $H$ is of size $O(\log^c n)$ and first-order definable from $G, E^+$ and $\ans$. Hence, by Theorem \ref{theorem:nl-fo}, the transitive closure of $H$ can be defined by a first-order formula. 

The transitive closure relation of $G' \df (V, E \cup E^+)$ can now be constructed from the transitive closures of $G$ and $H$. To this end observe that the transitive closure of $H$ equals the transitive closure relation of $G'$ restricted to $\Vaff$: it accounts for all paths from a node $u \in \Vaff$ to another node $v \in \Vaff$ that may use both newly inserted edges and edges that are already present in $G$.  For this reason, every path $\rho$ in $G'$ consists of three consecutive subpaths $\rho_1 \rho_2 \rho_3 = \rho$, where $\rho_1$ and $\rho_3$ are defined as the maximal subpaths of $\rho$ that do not rely on edges from $E^+$. These subpaths already exist in $G$ and are represented in~$\ans$. The subgraph $\rho_2$ by definition starts and ends at nodes from $\Vaff$, so its existence is given by the transitive closure relation of $H$.

Hence, the transitive closure of $G'$ can be defined by the formula $\varphi(s,t) \df \ans(s,t) \vee \exists x_1 \exists x_2 \: \big(\ans(s,x_1) \wedge \RelName{TC}_H(x_1,x_2) \wedge \ans(x_2,t)\big)$. 
\end{proof}

\begin{theorem}\label{theorem:bulk:mso:ureach}
  Reachability on undirected graphs can be maintained in $\DynFOar$ under changes of size $\log^c n$, for every $c \in \N$.
\end{theorem}
\begin{proof}
The dynamic program from \cite{DongS98} that maintains undirected reachability in \DynFO under single-edge changes uses, in addition to the transitive closure relation of the input graph, two binary auxiliary relations that represent a directed spanning forest of the input graph and its transitive closure, respectively.
We show that these relations can still be maintained in \DynFOar under changes of $\log^c n$ many edges, for fixed $c \in \N$.

Recall that it suffices to treat insertions and deletions independently, as they can be handled subsequently by a dynamic program.

For edge insertions, the construction idea is very similar to the proof of Theorem \ref{theorem:bulk:mso:reach}. We define a graph $H$, where nodes correspond to connected components of the input graph that include an affected node, and edges indicate that some inserted edge connects the respective connected components. As this graph is of polylogarithmic size, thanks to Theorem~\ref{theorem:nl-fo} we can express a spanning forest for $H$ and its transitive closure in \FOar, which is sufficient to update the respective relations for the whole input graph.

In the case of edge deletions, the update formulas need to replace deleted spanning tree edges, whenever this is possible.
Our approach is very similar to the case of edge insertions. The spanning tree decomposes into polylogarithmically many connected components when edges are deleted. These components can be merged again if non-tree edges exist that connect them, and these edges become tree edges of the spanning forest.
For a correspondingly defined graph of polylogarithmic size we can again define a spanning forest and its transitive closure, and from this information select the new tree edges.

We explain both cases in more detail. Let $G = (V,E)$ be the undirected input graph of size $n$ with transitive closure $\ans$, and let $S$ and $\RelName{TC}_S$ be a directed spanning forest for $G$ and its transitive closure, respectively.
Suppose that a set $E^+$ of size at most $\log^c n$ is inserted. 
We define a graph $H$ as follows. It contains a node $v \in V$ if (1) $v$ is affected, that is, if $E^+$ contains an edge of $v$, and (2) $v$ is the smallest affected node in its connected component of $G$ with respect to $\leq$. 
It contains an edge $(u,v)$ if $(u',v') \in E^+$ for some nodes $u', v'$ with $(u,u') \in \ans$ and $(v,v') \in \ans$, so, if the connected components of $u$ and $v$ are connected by an inserted edge.
The graph $H$ is easily seen to be $\FO$-definable using $\ans$.
Because $H$ is of polylogarithmic size with respect to $n$ and a spanning forest of a graph can be computed\footnote{For example the breadth-first spanning forest with the minimal nodes of each component, with respect to $\leq$, as roots can be computed with the inductive counting technique due to Immerman and Szelepcs{\'{e}}nyi \cite{Immerman88, Szelep1987}.} in \NL, we can define a spanning tree $S_H$ as well as its transitive closure $\RelName{TC}_{S_H}$ in \FOar, thanks to Theorem~\ref{theorem:nl-fo}.

The update formulas define updated auxiliary relations for the graph $G' = (V, E \cup E^+)$ as follows.
Intuitively, an edge $(u,v) \in S_H$ means that the connected components of $u$ and $v$ in $G$ shall be connected in $G'$ directly by a new tree edge. There might be several edges in $E^+$ that may serve this purpose, and we need to choose one of them. So, an edge $(u',v') \in E^+$ becomes part the updated spanning forest if there is an edge $(u,v) \in S_H$ such that $u'$ and $u$ as well as $v'$ and $v$ are in the same connected component of $G$, respectively, and $(u',v')$ is the lexicographically minimal edge with these properties. This is clearly \FOar-expressible using the old auxiliary relations. 
The old tree edges from $S$ are taken over to the updated version, although some directions need to be inverted, if for a newly chosen tree edge $(u',v')$ the node $v'$ was not the root of the directed spanning tree of its connected component. First-order formulas that determine which edges need to be reversed and that provide the adjusted transitive closure for the components of the spanning forest are given in \cite{DongS98}. The relation $\RelName{TC}_S$ is updated by combining this information with $\RelName{TC}_{S_H}$.

We note that $\ans$ is first-order expressible from $\RelName{TC}_S$.
In conclusion, all auxiliary relations can be updated in \FOar.

Now suppose that a set $E^-$ of at most $\log^c n$ edges is deleted.
Let $S'$ be the spanning forest that results from $S$ after all tree edges from $E^-$ are removed, and let $\RelName{TC}_{S'}$ be its transitive closure, which is easily \FO-expressible from $\RelName{TC}_S$.
Similarly as above we define a graph $H$, with nodes being the minimal affected nodes in a weakly connected component of $S'$, which are connected by an edge if the respective weakly connected components of $S'$ are connected by some edge from $E \setminus E^-$.
The same way as above, \FOar formulas can define a spanning forest and its transitive closure for $H$ and then use this information to define a spanning forest and its transitive closure for the changed graph $G' = (V, E \setminus E^-)$.
\end{proof}
 
\section{Handling Polylog Changes with \texorpdfstring{\DynFOParity}{\DynClass{FO+Mod2}}}
\label{section:isolation}
\def\shrinkage{2.1mu}
\def\vecsign{\mathchar"017E}
\def\dvecsign{\smash{\stackon[-1.95pt]{\mkern-\shrinkage\vecsign}{\rotatebox{180}{$\mkern-\shrinkage\vecsign$}}}}
\def\dvec#1{\def\useanchorwidth{T}\stackon[-4.2pt]{#1}{\,\dvecsign}}
\newcommand{\bidirect}[1]{\dvec{#1}}
\newcommand{\shift}[1]{\ensuremath{#1^{\uparrow}}}

While we have seen, in the last section, that reachability for directed graphs can be maintained under edge insertions of polylogarithmic size, a matching result for edge deletions is still missing.
Two intermediate results were shown in \cite{DattaMVZ18}, building on the work of \cite{Hesse03Reach}: reachability can be maintained in \DynFOar under insertions and deletions that affect $\frac{\log n}{\log \log n}$ nodes, and in \DynFOMaj under insertions and deletions of polylogarithmically many edges.

In this section, we adapt the proof of the latter result, also using ideas that appear in~\cite{DattaHK14}, and show that reachability can be maintained in \DynFOParityar under edge insertions and deletions of polylogarithmic size for classes of directed graphs for which non-vanishing weight assignments can be computed in \AC, that is, by polynomial-size circuits of polylogarithmic depth.
This is possible for example for planar graphs \cite{TewariV12}, as well as for graphs with bounded treewidth, as we show towards the end of this section.

We start by giving the necessary definitions regarding isolating and non-vanishing weight assignments. 

\subsection{Isolating and non-vanishing weights}
\label{section:isolation:prelims}
A \emph{weighted} directed graph $(G, w)$ consists of a graph $G = (V,E)$ and a \emph{weight assignment} $w \colon E \to \Z$ that assigns an integer weight $w(e)$ to each edge $e \in E$. 
The weight assignment $w$ is \emph{bounded} by a function $f(|V|)$ if $w$ assigns only weights from the interval $[-f(|V|),f(|V|)]$.

The weighted graph $(G, w)$ is \emph{min-unique} if (1) $w$ only gives positive weights to the edges $E$, and (2) if some path from $s$ to $t$ exists, for some pair $s,t$ of nodes, then there is a unique path from $s$ to $t$ with minimum weight under $w$.
Here, the weight of a path (and in general every sequence of edges) is the sum of the weights of its edges.
If $(G, w)$ is min-unique, we say that $w$ \emph{isolates} (minimal paths in) $G$.

Define $\bidirect{G} = (V, \bidirect{E})$ to be the \emph{bidirected extension} of $G$, where $\bidirect{E} \df \{(u,v), (v,u) \mid (u,v) \in E\}$.
A weight assignment $w$ is \emph{skew-symmetric} if $w(u,v) = -w(v,u)$ for all $(u,v) \in \bidirect{E}$. It has \emph{non-zero circulation} if the weight of every simple directed cycle in $\bidirect{G}$ is non-zero (here, a cycle is \emph{simple} if no node occurs twice).

From polynomially bounded non-zero circulation weights for $\bidirect{G}$ we can easily compute isolating weights for $G$.

\begin{restatable}{lemma}{nonzeroIsolating}
\label{lemma:nonzero-isolating}
Let $G = (V,E)$ be a graph with $n$ nodes, and let $w$ be skew-symmetric non-zero circulation weight assignment for $\bidirect{G}$, which is bounded by $n^k$ for some $k \in \N$.
Then $w'$ with $w'(e) = w(e) + n^{k+2}$ for every $e \in E$ isolates $G$. 
\end{restatable}
\begin{proof}
All weights in $w'$ are clearly positive. It remains to show that $w'$ isolates minimal paths in $G$.
Assume, towards a contradiction, that there are two different $s$-$t$-paths $\rho_1, \rho_2$ with the same minimal weight under $w'$ in $G$, for some nodes $s$ and $t$. Without loss of generality, they are both simple paths, as otherwise they cannot be minimal.
Let $u$ be the last node visited by both paths before they differ for the first time, and let $v$ be the first node after $u$ that is visited by both paths.
Let $\rho_1^{uv}, \rho_2^{uv}$ be the subpaths in $\rho_1, \rho_2$ from $u$ to $v$, respectively.
If these subpaths have different weights, say, $w'(\rho_1^{uv}) < w'(\rho_2^{uv})$, then we can replace $\rho_2^{uv}$ by $\rho_1^{uv}$ in $\rho_2$ and get a lighter path, contradicting the assumption that both $\rho_1$ and $\rho_2$ are paths with minimal weight.
So, $w'(\rho_1^{uv}) = w'(\rho_2^{uv})$ needs to hold. Then also $w(\rho_1^{uv}) = w(\rho_2^{uv})$ holds, because $w(\rho)$ and $w'(\rho)$ differ by a multiple of $n^{k+2}$ for any path $\rho$, and the difference between $w(\rho)$ and $w(\rho')$ is at most $n^{k+1}$, for simple paths $\rho$ and $\rho'$.
So, $w'$ cannot compensate weight differences under $w$.
But then the concatenation of $\rho_1^{uv}$ and the reverse of $\rho_2^{uv}$ is a simple cycle in $\bidirect{G}$ with weight $w(\rho_1^{uv}) - w(\rho_2^{uv}) = 0$, contradiction the assumption that $w$ has non-zero circulation. 
\end{proof}

We explain how (families of) polynomially bounded weight assignments for graphs are represented in relational structures. 
Let $V$ be the node set of a weighted graph of size $n$. We identify $V$ with the set $\{0, \ldots, n-1\}$ of numbers according to the given linear order $\leq$. A tuple $(a_1, \ldots, a_k)$ of nodes then represents the number $\sum_{i=1}^k a_i n^{i-1}$.
A (partial) function $f \colon V^k \to V^\ell$ is represented as a $(k+\ell)$-ary relation $F$ over $V$, such that for each $\tpl a \in V^k$ there is at most one $\tpl b \in V^\ell$ with $(\tpl a,\tpl b) \in F$.
We say that $f$ is \FOar-definable if $F$ is defined by an $\FOar$ formula $\psi(\tpl x, \tpl y)$, where $\tpl x = x_1, \ldots, x_k$ and $\tpl y = y_1, \ldots, y_\ell$. 
An $\FOar$ formula $\psi(\tpl z, \tpl x, \tpl y)$, with $\tpl z = z_1, \ldots, z_m$, defines a family $\{ f(\tpl c) \colon V^k \to V^\ell \mid \tpl c \in V^m \}$ of functions.
 
\subsection{Maintaining Reachability in weighted graphs}
\label{section:isolation:reach}
We now state and prove the main result of this section.

\begin{theorem}\label{theorem:reach-mod2}
Let $\calG$ be a class of graphs for which polynomially bounded skew-symmetric non-zero circulation weights can be computed in $\AC$.
Then, reachability for graphs in $\calG$ is in \DynFOParityar under changes of size $\log^c n$, for every $c \in \N$.
\end{theorem}

Although this result leaves open whether reachability can be maintained in \DynFOar under polylogarithmically many edge changes, note that, in light of Corollary~\ref{cor:barrier}, it gives a tight upper bound for the size of changes that can be handled in \DynFOParityar.

We outline the proof strategy, which closely follows the strategy from \cite{DattaMVZ18}.
Suppose, we are given a weighted directed graph $(G, w)$ where $G = (V,E)$ is a graph with $n$ nodes and $w$ is an isolating weight assignment. We represent this weighted graph by an $n \times n$ matrix $A_{(G, w)}(x)$ as follows: if $(u,v) \in E$, then the $u$-$v$-entry of $A_{(G, w)}(x)$ is $x^{w(u,v)}$, where $x$ is a formal variable, otherwise the $u$-$v$-entry is $0$. 

The matrix $D \df \sum_{i=0}^\infty (A_{G_w}(x))^i$ is a matrix of formal power series in the formal variable~$x$, and from an $s$-$t$-entry $\sum_{i=0}^\infty c_i x^i$ of this matrix we can read the number $c_i$ of paths from $s$ to $t$ with weight $i$.
Our goal is to determine the coefficients $c_i$ modulo $2$, for all $i$ up to some polynomial bound. From this information we can deduce whether there is a path from $s$ to $t$ in $G$: as $w$ isolates minimal paths in $G$, if there is some path from $s$ to $t$, then there is a unique path with minimal weight, which means that for the weight $\ell$ of this path we have $c_\ell \equiv 1 \pmod 2$. 
Otherwise, if no path from $s$ to $t$ exists, $c_\ell \equiv 0 \pmod 2$ for all $i$.

We use the following insights to actually compute and update the coefficients $c_i$. Notice that the matrix $D$ is invertible over the ring of formal power series (see \cite{DattaMVZ18} and its full version \cite{DattaMVZ18a}) and can be written as 
$D = (I - A_{(G, w)}(x))^{-1}$, where $I$ is the identity matrix. 

So, we need to compute and update the inverse of a matrix. This cannot be done effectively for matrices of inherently infinite formal power series. For this reason we compute $D$ only approximately. A \emph{$b$-approximation} $C$ of $D$, for some $b \in \N$, is a matrix of formal polynomials that agrees with the entries of $D$ on the low-degree coefficients $c_i$ for all $i \leq b$. This precision is preserved by the matrix operations we use, see \cite[Proposition~14]{DattaMVZ18}. Note that it is sufficient to maintain an approximation of $D$, as for a weighted graph with polynomially bounded weights the maximal possible weight $w_\text{max}$ of a minimal path is bounded by a polynomial, and thus only the coefficients $c_i$ with $i \leq w_\text{max}$ are relevant.

To update the matrix inverse, we employ the Sherman-Morrison-Woodbury identity (cf.~\cite{HendersonS81}).
This identity states that when updating a matrix $A$ to a matrix $A + \Delta A$, with $\Delta A$ writeable as matrix product $UBV$, the inverse of $A$ can be updated as follows:  
\[(A + \Delta A)^{-1} = (A + UBV)^{-1} = A^{-1}-A^{-1}U(I+BVA^{-1}U)^{-1}BVA^{-1}.\]

When $\Delta A$ has only $k$ non-zero rows and columns, there is a decomposition $UBV$ where $B$ is a $k \times k$ matrix.

The right-hand side can be computed in $\FOParityar$ for $k \df \log^c n$. To see this, we observe that also $I+BVA^{-1}U$ is a $k \times k$ matrix. Computing the right-hand side now requires multiplication and  iterated addition of polynomials over $\Z$ as well as the computation of the inverse of a $k \times k$ matrix. As all computations are done modulo $2$, this is indeed possible in \FOParityar for (matrices of) polynomials with polynomial degree using results of \cite{HealyV06}. We provide more details later.

As we work with isolating weight assignments, our update routines also need to assign weights to changed edges such that the resulting weight assignment is again isolating. We show that this can be done if we start with (slightly adjusted) non-zero circulation weights. Using Theorem~\ref{theorem:oneshotmuddling} we can assume that such an assignment is given, and that we only need to update the weights once.

\begin{proofsketchof}{Theorem~\ref{theorem:reach-mod2}}

Let $c$ be arbitrary. Thanks to Theorem~\ref{theorem:oneshotmuddling} it suffices to show that there is a $d \in \N$ such that reachability  is $(\AC^d,1)$-maintainable by a dynamic program $\prog$ under changes of size $\log^{c+d} n$. Let $d' \in \N$ be such that polynomially bounded skew-symmetric non-zero circulation weights for graphs from $\calG$ can be computed in $\AC^{d'}$, and set $d\df \max(2, d')$. 

Let $G = (V,E)$ be a graph with $n$ nodes. Let $u$ be skew-symmetric non-zero circulation weights for $G$ and let $n^k$ be the polynomial bound on the weights. Further, let $w$ be the weight assignment that gives weight $n^{k+2} + u(e)$ to each edge $e \in E$. Notice that $w$ is polynomially bounded by $n^{k+3}$ and isolates $G$ according to Lemma~\ref{lemma:nonzero-isolating}.

The $\AC^d$ initialization computes, as auxiliary information, the weightings $u$ and $w$ and an $n^{b}$-approximation~$C$ of $(I - A_{(G, w)}(x))^{-1} \bmod 2$, that is, a matrix of formal polynomials in $x$ that agree with the formal power series in $(I - A_{(G, w)}(x))^{-1} \bmod 2$ on the coefficients up to degree $n^{b}$. Here, $b \in \N$ is a constant to be determined later.

When changing $G$ via a change $\Delta E$ with deletions $E^-$ and insertions $E^+$, the dynamic program $\prog$ handles deletions and insertions subsequently:
\begin{enumerate}[(1)]
 \item Handling of deletions:
    \begin{enumerate}[(a)]
      \item Define isolating weights $w^-$ for $G^- \df (V, E \setminus E^-)$. The weights $w^-$ will differ from $w$ only for the at most $\log^{c+d} n$ edges in $E^-$.
      \item Compute an $n^b$-approximation of \mbox{$(I - A_{(G^-, w^-)}(x))^{-1} \bmod 2$} using the existing $n^b$-approximation  of $(I - A_{(G, w)}(x))^{-1} \bmod 2$.
    \end{enumerate}
  \item Handling of insertions:
    \begin{enumerate}[(a)]
     \item Define a family $W^{-/+}$ of weightings such that one member of the family is isolating for $G^{-/+} \df (V, (E \setminus E^-) \cup E^+)$. All weightings of the family will differ from $w^-$ only for the at most $\log^{c+d} n$  edges in $E^+$.
     \item Compute an $n^b$-approximation of \mbox{$(I - A_{(G^{-/+}, w^{-/+})}(x))^{-1} \bmod 2$} using the existing $n^b$-approximation of $(I - A_{(G^-, w^-)}(x))^{-1} \bmod 2$ for all members $w^{-/+}$ of $W^{-/+}$.
    \end{enumerate}
\end{enumerate}

We first explain Steps (1a) and (2a) in more detail. For computing the isolating weights $w^-$, the program proceeds as follows. Skew-symmetric non-zero circulation weights $u^-$ for $G^-$ are obtained from the non-zero circulation weights $u$ for $G$ by setting the weight of deleted edges $e \in E^-$ to $0$. As $u^-$ gives the same weight to all simple cycles in $G^-$ as $u$ gives to these cycles in $G$, it has non-zero circulation. Now, the weight assignment $w^-$ defined by $n^{k+2} + u^-(e)$ is isolating for $G^-$ due to Lemma~\ref{lemma:nonzero-isolating}, and differs from $w$ only for edges in $E^-$.

Computing the isolation weights for insertions is more challenging. In Lemma~\ref{lemma:update-weights} below we show that from $G^-$, its transitive closure, and a set $E^+$ of edges of polylogarithmic size one can \FOar-define a family $W^{-/+}$ of weight assignments such that one of these assignments is isolating for $G^{-/+}$. %

For both Steps (1b) and (2b), the inverse of a matrix of polynomials over $\Z_2$ of polynomial degree needs to be updated after changing polylogarithmically many entries (i.e.\ entries corresponding to $E^-$ and $E^+$, respectively). Inverses can be updated under such changes in \FOParityar due to Lemma~\ref{lemma:inverseapproximation_update} (see below) and the observation that changes $\Delta A$ of size $\log^{c+d} n$  to such a matrix can be decomposed into $UBV$ as required by Lemma \ref{lemma:inverseapproximation_update}, see Lemma 7 in \cite{DattaMVZ18}. For Step (2b) this is done in parallel for all members of $W^{-/+}$.

For checking whether there is a path from $s$ to $t$ after the change $\Delta E$ to $G$, the dynamic program checks whether there is a member of $W^{-/+}$ such that the $s$-$t$-entry of $(I - A_{(G^{-/+}, w^{-/+})}(x))^{-1} \bmod 2$ is non-zero. Since one member of $W^{-/+}$ is isolating, a path will be discovered this way.
\end{proofsketchof}

In the remainder of this subsection we show how inverses for matrices of polynomials can be updated under changes of polylogarithmic size, and how weights for inserted edges can be found.

The following lemma is obtained using the same techniques as in \cite{DattaMVZ18}. Here, $\Z_2[[x]]$ denotes the ring of formal power series with coefficients from $\Z_2$, and $\Z_2[x]$ denotes its subring that consists of all finite polynomials.

\begin{restatable}{lemma}{lemmaInverseapproximationUpdate}
\label{lemma:inverseapproximation_update}
    Suppose $A \in \Z_2[[x]]^{n \times n}$ is invertible over $\Z_2[[x]]$, and $C \in \Z_2[x]^{n \times n}$ is an $m$-approximation of $A^{-1}$. If $A + \Delta A$ is invertible over $\Z_2[[x]]$ and $\Delta A$ can be written as $UBV$ with $U \in \Z_2[x]^{n \times k}, B \in \Z_2[x]^{k \times k},$ and $V \in \Z_2[x]^{k \times n}$, then
      \[(A + \Delta A)^{-1} \approx_m C-CU(I+BVCU)^{-1}BVC\]
    Furthermore, if $k \leq \log^c n$ for some fixed $c$ and all involved polynomials have polynomial degree in $n$, then the right-hand side can be defined in \FOParityar from $C$ and $\Delta A$.
\end{restatable}
\begin{proofsketch}
   The correctness of the equation can be proved exactly as in Proposition 14 in \cite{DattaMVZ18} (there, this is proved for $\Z[[x]]$ instead of $\Z_2[[x]]$). 
   
   We argue that the right-hand side can be defined in $\FOParityar$. The involved matrix additions and multiplications modulo $2$ can easily be expressed in \FOParityar, see \cite{HealyV06}. It remains to explain how the inverse of the $\log^c n \times \log^c n$ matrix $I+BVCU$ can be found.

    To this end, recall that the $i$-$j$-entry of the inverse of a matrix $D$ is equal to $(-1)^{i+j} \frac{\det D_{ji}}{\det D}$, where $D_{ji}$ is obtained from $D$ by removing the $j$-th row and the $i$-th column.
    
    So, it is sufficient to show that the determinant of a $\log^c n \times \log^c n$ matrix of polynomials with polynomial degree can be expressed modulo $2$.  In \cite[Lemma~15]{DattaMVZ18} it was shown that such a determinant can be expressed in \FOMajar, by observing that one only needs to be able to express the sum of polynomially many polynomials and the product of $\log^c n$ many polynomials. This observation is still valid for computing the determinant modulo $2$ in \FOParityar. Both kind of computations are possible modulo $2$ in \FOParityar as well \cite{HealyV06}.
\end{proofsketch}

\begin{restatable}{lemma}{lemmaUpdateWeights}
\label{lemma:update-weights}
Let $G = (V,E)$ be a graph and let $n = |V|$. Further, let $w$ be a polynomially bounded isolating weight assignment for $G$, and let $E^+$ be a set of $\bigO(\log^c n)$ edges that is disjoint from $E$, for some $c \in \N$.
Then there is a family $W'$ of polynomially many polynomially bounded weight assignments such that
\begin{enumerate}
 \item $W'$ is \FOar-definable from $G$, $\Reach(G)$, $E^+$ and $w$,
 \item all $w' \in W'$ agree with $w$ on $E$,
 \item at least one $w' \in W'$ is isolating for $(V, E \cup E^+)$.
\end{enumerate}
\end{restatable}

The proof works along the following lines. We use the approach from \cite{KallampallyT16} to obtain weights for the inserted edges with the following idea: if there is an $s$-$t$-path that uses at least one inserted edge from $E^+$, then there is a unique minimal path under all $s$-$t$-paths that use at least one such edge, where we ignore the weight of the paths that is contributed by edges from $E$.
We multiply these constructed weights for the edges from $E^+$ by a large polynomial to ensure that the combined weight assignment with the existing weights for edges in $E$ is isolating for the graph $(V, E \cup E^+)$.

The approach from \cite{KallampallyT16} does not lead to polynomially bounded weights in the size of the graph it is used for.
We construct them for a graph with $N = \bigO(\log^c n)$ many nodes, and although they are not polynomially bounded in $N$, they are in $n$.

We now get to the details of the construction.
In the following, we consider graphs with two sets of edges. An \emph{adorned} graph $G = (V,E,F)$ has, besides the set $E$ of \emph{real} edges, a further set $F$ of \emph{fictitious} edges, which is not necessarily disjoint from $E$.
For each pair $s, t$ of nodes, let $\calP'_{s,t}$ be the set of $s$-$t$-paths in $G$ that use at least one real edge $e \in E$ and arbitrarily many fictitious edges $e' \in F$.
Let $\calP_{s,t}$ be the set of edge sequences that result from $\calP'_{s,t}$ by removing the fictitious edges from the paths.

We say that a weight assignment $w$ \emph{real-isolates} $G$, if (1) it maps each real edge $e \in E$ to a positive integer, and (2) each non-empty $\calP_{s,t}$ has a unique minimal element under $w$.
In the following, we will need a stronger property. We say that $w$ \emph{strongly} real-isolates $G$, if in addition for each pair $\calP_{s,t}$ and $\calP_{s',t'}$ of non-empty sets with $(s,t) \neq (s',t')$ the unique minimal elements of $\calP_{s,t}$ and $\calP_{s',t'}$ have different weights under $w$.

The following lemma can be proved along the lines of \cite{KallampallyT16}. The proof is given at the end of this subsection. 

\begin{lemma}\label{lemma:compute-isolating}
There is a constant $\beta \in \N$ such that for every natural number $N$ and every adorned graph $G = (V,E,F)$ with $V= \{1, \ldots, N\}$ there is a sequence $\tpl p = p_1, p_2, \ldots, p_{\log N}$ of primes, each consisting of at most $(\beta-2) \log N$ bits, such that the weight assignment 
\[w_{\tpl p}(e) = \begin{cases}
  \sum_{j = 1}^{\log N} N^{\beta(\log N - j)}(w_0(e) \bmod p_j) & e \in E \\
  0 & e \in F
\end{cases}\]
strongly real-isolates $G$. Here, $w_0(u,v) \df 2^{{(N+1)}u + v}$.
\end{lemma}

Using this lemma, we can prove Lemma~\ref{lemma:update-weights}.
\begin{proofof}{Lemma \ref{lemma:update-weights}}
Let $\Vaff \subseteq V$ be the set of nodes with edges in $E^+$.
We construct an adorned graph $H = (V_H,E_H,F_H)$ with node set $V_H \df \Vaff$ as follows. The set $E_H$ of real edges is $E_H \df E^+$, and the set $F_H$ of fictitious edges is $F_H \df \{ (u,v) \mid u,v \in V_H, (u,v) \in \Reach(G)\}$. 
So, a fictitious edge $(u,v)$ of $H$ represents the existence of a $u$-$v$-path in $G$. 

Let $\beta, \tpl p$ and $w_{\tpl p}$ be as promised to exist for $H$ by Lemma~\ref{lemma:compute-isolating}. Further, let $n^k$ be the upper bound on the weights of $w$.
We define the weight assignment $w'$ for $G' \df (V, E \cup E^+)$ as follows.
 \begin{itemize}
  \item $w'(e) = w(e)$ for all $e \in E$,
  \item $w'(e) = n^{k+2} \cdot w_{\tpl p}(e)$ for all $e \in E^+ = E_H$.
 \end{itemize}

We show that $w'$ isolates $G'$ first, afterwards we show that $w'$ is a member of an \FOar-definable family of weightings.

For showing that $w'$ isolates $G'$ suppose, towards a contradiction, that there are two lightest simple $s$-$t$-paths $\pi, \rho$ in $G'$ with respect to $w'$, for some nodes $s$ and $t$.
Let $\pi_1\pi_2\pi_3 = \pi$ and $\rho_1\rho_2\rho_3 = \rho$ be the subpaths of $\pi$ and $\rho$ such that edges from $E^+$ are only used in $\pi_2$ and $\rho_2$ and those subpaths are minimal with that property. 
Notice that both $\pi_2$ and $\rho_2$ are non-empty, as otherwise $\rho$ and $\pi$ are also lightest paths in $G$ with respect to $w$, contradicting the assumption that $w$ isolates $G$.
Let $\pi'$ and $\rho'$ be the paths in $H$ that correspond to $\pi_2$ and $\rho_2$, where subpaths of $\pi_2$ and $\rho_2$ are replaced by fictitious edges.
We consider two cases.

If $w_{\tpl p}(\pi') \neq w_{\tpl p}(\rho')$, then the total weight of $\pi$ and $\rho$ contributed by the edges from $E^+$ differs by at least $n^{k+2}$. As the total weight contributed by the remaining edges is upper-bounded by $n^{k+1}$, we have that $w'(\pi) \neq w'(\rho)$, the desired contradiction.

Thus assume, without loss of generality, that $w_{\tpl p}(\pi') = w_{\tpl p}(\rho')$. We can assume that both paths $\pi'$ and $\rho'$ are lightest paths in $H$: if, say, $\pi'$ is not a lightest path, then we can replace $\pi_2$ in $\pi$ by a path that uses lighter edges from $E^+$, leading to an overall lighter path by the argument of the previous case. 
Then, because $w_{\tpl p}$ strongly isolates $H$, the edges from $E_H = E^+$ used in $\pi'$ and $\rho'$ must be equal, and the same is true for $\pi_2$ and $\rho_2$. These edges must also be used in the same order, as otherwise a path with fewer edges from $E^+$ exists, which by the argument of previous case is lighter than both $\pi$ and $\rho$.
Because $\pi$ and $\rho$ are different paths, there must be subpaths $\pi^*$ and $\rho^*$ that consist only of edges from $E$ and are both simple $u$-$v$-paths, for some nodes $u$ and $v$. 
As $w$ is isolating for $G$, not both subpaths can be lightest $u$-$v$-paths in $G$. Say, $\rho^*$ is not such a lightest path. If we replace $\rho^*$ in $\rho$ by the lightest $u$-$v$-path in $G$, we obtain a path that is lighter than $\rho$, as $w'$ agrees with $w$ on $E$. 
So, $\rho$ is not a lightest $s$-$t$-path in $G'$ with respect to $w'$, the desired contradiction.
It follows that $w'$ is isolating for $G'$.

The weight assignment $w_{\tpl p}$ is clearly $\FOar$-definable from $G, \Reach(G), E^+, w$ and $\tpl p$, as $H$ is $\FOar$-definable, the involved numbers consist of at most polylogarithmically many bits, and \FOar can express the necessary arithmetic on numbers of that magnitude (see \cite[Theorem~5.1]{HesseAB02}).
The sequence $\tpl p$ consists of $\bigO(\log \log n)$ many primes (as $H$ is of size polylog) which in turn are represented by $\bigO(\log \log n)$ many bits, because $H$ has only polylogarithmic size in $n$. 
So, $\tpl p$ can be represented by a tuple of nodes from $V$, and it follows that a family $W'$ of weight assignments with $w' \in W'$ is $\FOar$-definable from $G, \Reach(G), E^+$ and $w$.
\end{proofof}

It remains to prove Lemma~\ref{lemma:compute-isolating}. As already mentioned above, the proof closely follows \cite{KallampallyT16}. It also uses the following lemma.

\begin{lemma}[\cite{FKS}, see also \cite{PTV}]\label{lem:FKS}
For every constant $c$ there is a constant $c_0$ such that for every set
$S$ of $m$ bit integers with $|S| \leq m^c$ the following holds:
There is a $c_0\log{m}$ bit prime number $p$ such that for any $x \neq y \in S$ it holds that $x \not\equiv y \pmod{p}$.
\end{lemma}

\begin{proofof}{Lemma~\ref{lemma:compute-isolating}}
Let $\beta$ be some constant to be defined later, and let $G = (V,E,F)$ be an adorned graph with $V = \{1, \ldots, N\}$. For each pair $s,t$ of nodes, let $\calP_{s,t}$ be defined as above. We call an element from $\calP_{s,t}$ a \emph{real-partial} $s$-$t$-path. Let $\calP^i_{s,t}$ denote the real-partial $s$-$t$-paths that consist of at most $2^i$ (real) edges, for $i \leq \log N$.

We first describe the proof idea. 
We will choose primes $\tpl p = p_1, \ldots, p_\ell$ with $\ell = \log N$, such that each prime $p_j$ is smaller than $N^{\beta-2}$. Then, each weight assignment $w_j(e) = w_0(e) \bmod p_j$ gives a total weight of less than $N^\beta$ to every element from a set $\calP_{s,t}$.
That means we can consider the weight that $w_{\tpl p}$ assigns to an edge or to an element from $\calP_{s,t}$ to be a number with radix $N^\beta$, and the $j$-th most significant digit is given by the function $w_j$.
So, whether an edge (or a real-partial $s$-$t$-path) is lighter than another edge (or another real-partial $s$-$t$-path) under $w_{\tpl p}$ is determined by the smallest $j$ such that $w_j$ assign a different weight to the edges (or real-partial $s$-$t$-paths).

We choose the primes $\tpl p$ inductively in a way such that the function \[w_i(e) \df \sum_{j = 1}^{i} N^{\beta(\log N - j)}(w_0(e) \bmod p_j)\] strongly real-isolates the sets $\calP^i_{s,t}$.
Then, the claim of the lemma follows.

We start with $i = 0$. Each real-partial $s$-$t$-path in a set $\calP^0_{s,t}$ has at most one real edge, so the number of all these elements, summed over all pairs $s,t$, is bounded by $N^4$. 
So, by Lemma~\ref{lem:FKS} there is a constant $c_1$ and a $c_1 \log N$ bit prime $p_1$ such that $w_1$ strongly real-isolates the sets $\calP^0_{s,t}$.

For the inductive step, let the claim hold for some $i$. Each set $\calP^{i+1}_{s,t}$ consist of real-partial paths that are the concatenation from one real-partial path from $\calP^{i}_{s,v}$ and one real-partial path $\calP^{i}_{v,t}$, for some node $v$. 
Also, the lightest element in $\calP^{i+1}_{s,t}$ under $w_i$ consists of two real-partial paths that are lightest under $w_i$ among the sets $\calP^{i}_{s,v}$ and $\calP^{i}_{v,t}$.
So, each set $\calP^{i+1}_{s,t}$ can contain at most $N$ real-partial paths that are lightest under $w_i$, and there are at most $N^3$ of these elements among all sets $\calP^{i+1}_{s,t}$ in total.
So, by Lemma~\ref{lem:FKS} there is a constant $c_2$ (which is the same for all applications of the inductive step) and a $c_2 \log N$ bit prime $p_{i+1}$ such that $w_{i+1}$ strongly real-isolates the sets $\calP^{i+1}_{s,t}$.

The proof is finished by choosing $\beta$ to be $\max(c_1,c_2)+2$.
\end{proofof}
 
\subsection{Computing weights for bounded-treewidth graphs}
\label{section:isolation:btw}
\newcommand{\tparent}{\text{parent}}

In this section, we show that isolating weights for graphs of bounded treewidth can be computed in \LOGSPACE. As an immediate consequence, reachability can be maintained for such graphs under changes of polylogarithmic size.

A \emph{tree decomposition} $\calT = (T,\calB)$ of a graph $G= (V,E)$ consists of a (rooted, directed) tree $T = (I,F,r)$, with (tree) nodes $I$, (tree) edges $F$, a distinguished root node $r \in I$, and a function $\calB \colon I \rightarrow 2^V$ such that
\begin{enumerate}
 \item[(1)] the set $\{i \in I \mid v \in \calB(i)\}$ is non-empty for each node $v \in V$,
 \item[(2)] there is an $i \in I$ with $\{u,v\} \subseteq \calB(i)$ for each edge $(u, v) \in E$, and
 \item[(3)] the subgraph $T[\{i \in I \mid v \in \calB(i)\}]$ is connected for each node $v \in V$.
\end{enumerate}
We refer to the number of children of a node $i$ of $T$ as its \emph{degree}, and to the set $\calB(i)$ as its \emph{bag}. We denote the parent node of $i$ by $\tparent(i)$.
The \emph{width} of a tree decomposition is defined as the maximal size of a bag minus $1$. The \emph{treewidth} of a graph $G$ is the minimal width among all tree decompositions of $G$.
A tree decomposition is binary, if all tree nodes have degree at most $2$. Its depth is the length of a longest path from the root $r$ to a leaf of $T$. 
We inductively define the \emph{height} $h(i)$ of $i$ to be $1$ if $i$ is a leaf, and $h(i')+1$ if $i$ is an inner tree node and $i'$ is a child of $i$ with maximal height.
For a node $v \in V$ we denote by $B(v)$ the highest bag that contains $u$, and let $h(u) \df h(B(u))$. This bag $B(v)$ is well-defined for each node $v$ thanks to condition (3) of the definition of a tree decomposition.

We usually identify tree nodes $i$ and their bag $\calB(i)$, and use the above notions and measures directly for bags. We also abuse notation and write $B \in \calB$ if $B = \calB(i)$ for some tree node $i$.

As a first step we construct isolation weights for bounded treewidth graphs that additionally have bounded degree. Isolation weights for all graphs of bounded treewidth are provided afterwards.

\begin{restatable}{proposition}{propBtwDWeights}
\label{prop:btw-d-weights}
Let $c, d, k \in \N$ be fixed. Let $G = (V,E)$ be a graph with maximal degree $d$, and let $n$ be the number of its nodes.
Let $\calT$ be a binary tree decomposition of $G$ with width $k$ and depth at most $c \log n$. A polynomially bounded skew-symmetric weight assignment with non-zero circulation for $G$ can be computed in \LOGSPACE.
\end{restatable}

\begin{proof}
The idea for assigning weights is the following. We associate each edge with one bag of the tree decomposition, namely the highest bag that contains one endpoint of the edge. For each bag $B$, we denote the set of all edges that are associated with $B$ by $S(B)$. 
As the width of $\calT$ and the degree of $G$ are bounded by a constant, so is $|S(B)|$. 
An edge $e$ is assigned a weight that depends exponentially on the height of the bag $B$ it is associated with, and also exponentially on its position in some linear order on $S(B)$. For each cycle $C$ there is a unique highest bag $B_C$ that some of the cycle's edges is associated with. The idea for establishing non-zero circulation of $C$ is that its weight is dominated by the weight of the unique edge which (1) is associated with $B_C$ and (2) has largest index in the linear order on $S(B)$ among all edges of the cycle.
As the height of a bag is logarithmic in $n$ and $|S(B)|$ is bounded by a constant, the weight of every edge is polynomial in $n$.

We now proceed to the details. For each $e \in \bidirect{E}$, let $B_{e}$ be the (unique) highest bag that contains one of the end points of $e$. For a bag $B$, define the set $S(B)$ of its associated edges as  $S(B) \df \{e \in \bidirect{E} \mid B = B_{e}\}$. Observe that the sets $S(B)$ partition the set $\bidirect{E}$ of edges and that the size of $S(B)$ is bounded by a constant $\beta \df 2d(k +1)$, as each bag $B$ contains at most $k +1$ nodes and each node has degree at most $d$ in $G$ and therefore degree at most $2d$ in $\bidirect{G}$. For each $S(B)$ we fix an enumeration of its elements\footnote{As we devise an \LOGSPACE algorithm, we can assume the existence of a linear order on the input.}. Now, for each edge $e$, we set $h(e) = h(B_e)$ and $\ell(e) = i$, if $e$ is the $i$-th element in the enumeration of $S(B_e)$.

We set the weight $w(e)$ of an edge $e = (u,v)$ with $u \leq v$ to be $w(e) \df (4\beta \cdot 3^\beta + 2)^{h(e)} \cdot 3^{\ell(e)}$. The weight of an edge $(u,v)$ with $u > v$ is $w(u,v) = -w(v,u)$.
Notice that this weight assignment is polynomially bounded and skew-symmetric and can be computed in \LOGSPACE. We now show that it has non-zero circulation.

Let $\calC$ be any simple cycle in $\bidirect{G}$, and let $e_1, \ldots, e_m$ be an enumeration of its edges. Without loss of generality we assume that $e_1$ is the edge with the maximal weight among all edges in $\calC$. 
This edge is well-defined, as there is a unique highest bag $B$ such that $S(B)$ contains an edge of $\calC$, and the term $4\beta \cdot 3^\beta + 2$ is strictly greater than $3^{\ell(e)}$ for any value of $\ell(e)$.

We show $|w(e_1)| > |w(e_2) + \cdots + w(e_m)|$, which implies the claim.
Actually, we show that the weight of $w(e_1)$ exceeds the combined weight of all other edges $e$ that are either in $S(B)$ and have $\ell(e) < \ell(e_1)$ or are in $S(B')$ for some bag $B'$ below $B$ in the tree decomposition.
Note that there are $\sum_{h=1}^{h(e_1)-1} 2^{h(e_1)-h}$ many of those bags $B'$, each $S(B')$ contains at most $\beta$ edges, and the weight of each edge is upper bounded by $(4\beta\cdot3^\beta+2)^{h(B')}\cdot 3^\beta$. 
 \begin{align*}
|w(e_2) &+ \cdots + w(e_m)| \\ 
&< 
\sum_{i=1}^{l(e_1)-1}(4\beta \cdot 3^\beta+2)^{h(e_1)}\cdot 3^{i}
+\sum_{h=1}^{h(e_1)-1}2^{h(e_1)-h}\cdot \beta\cdot (4\beta\cdot3^\beta+2)^{h}\cdot 3^\beta \allowdisplaybreaks \\ 
&= \sum_{i=1}^{l(e_1)-1}(4\beta \cdot 3^\beta+2)^{h(e_1)}\cdot 3^{i}
+\sum_{h=1}^{h(e_1)-1}2^{h(e_1)}\cdot \beta\cdot (2\beta\cdot3^\beta+1)^{h}\cdot 3^\beta \allowdisplaybreaks \\ 
&= (4\beta \cdot 3^\beta+2)^{h(e_1)}\cdot\sum_{i=1}^{l(e_1)-1} 3^{i}
+2^{h(e_1)}\cdot \beta \cdot 3^\beta \cdot \sum_{h=1}^{h(e_1)-1} (2\beta\cdot3^\beta+1)^{h} \allowdisplaybreaks \\ 
 &=  (4\beta\cdot 3^\beta+2)^{h(e_1)}\cdot \frac{3^{l(e_1)}-3}{2} 
 + 2^{h(e_1)}\cdot \beta \cdot 3^\beta \cdot \frac{(2\beta \cdot 3^\beta + 1)^{h(e_1)}-(2\beta \cdot 3^\beta + 1)}{2\beta \cdot 3^\beta}  \\
&< (4\beta \cdot 3^\beta+2)^{h(e_1)}\cdot 3^{\ell(e_1)} = |w(e_1)| \qedhere
\end{align*}
\end{proof}

Non-zero circulation weights cannot only be computed in \LOGSPACE for bounded-treewidth graphs with bounded degree, as given by  Proposition~\ref{prop:btw-d-weights}, but also for all graphs with bounded treewidth.
Also, using a result of Elberfeld, Jakoby and Tantau \cite{ElberfeldJT10}, no tree decomposition needs to be given as input.

\begin{restatable}{theorem}{theoremBtw}
\label{theorem:btw}
Let $k \in \N$ be fixed and let $G = (V,E)$ be a graph with treewidth at most $k$. A polynomially bounded skew-symmetric weight assignment with non-zero circulation for $G$ can be computed in \LOGSPACE.
\end{restatable}

The idea for proving Theorem~\ref{theorem:btw} is as follows. From a given graph $G$ with treewidth at most $k$ we construct a graph $G'$ with treewidth and degree $\bigO(k)$ as well as a tree decomposition. The graph $G'$ basically results from a tree decomposition $\calT$ of $G$ by making a copy of a node $v$ for every bag of $\calT$ that contains $v$. These copies are connected by an edge if the corresponding bags in $\calT$ are.
Using Proposition~\ref{prop:btw-d-weights}, we obtain non-zero circulation weights for $G'$, and we show that they can be translated to non-zero circulation weights for~$G$.

\begin{proof}
Fix $k \in \N$ and let $G = (V,E)$ be a graph with treewidth at most $k$. Let $n$ be the size of $V$.
There are constants $c_1, c_2 \in \N$ that only depend on $k$ such that a binary tree decomposition $\calT = (T, \calB)$ of $G$ of width at most $c_1 k$ and depth at most $c_2 \log n$ can be computed in \LOGSPACE \cite{ElberfeldJT10}.
From $G$ and $\calT$ we construct a graph $G' = (V',E')$ as follows.
Let $V'$ be the set $V' \df \{v_B \mid B \in \calB, v \in B\}$ and let $E' \df \{(v_B,v_{B'}) \mid B' = \tparent(B)\} \cup \{(u_B,v_B) \mid (u,v) \in E, u \not\in \tparent(B) \text{ or } v \not\in \tparent(B)\}$.
So, we have one copy $v_B$ of a node $v \in V$ for each bag $B$ such that $v$ is contained in $B$. Two copies of a node are connected by an edge if they originate from adjacent bags in the tree decomposition, and there is an edge between two copies $u_B$ and $v_B$, originating from the same bag $B$, if $B$ is the highest bag of $\calT$ that contains both endpoints $u$ and $v$.

The degree of $G'$ is bounded by $c_1 k + 3$. The tree decomposition $\calT' = (T, \calB')$ that replaces each bag $B$ of $\calT$ by $\{v_B \mid v \in B\} \cup \{v_{\tparent(B)} \mid v \in \tparent(B)\}$ is a tree decomposition of $G'$ and has width at most $2 c_1 k + 1$. Furthermore, it is binary and has depth at most $c_2 \log n$.
So, by Proposition~\ref{prop:btw-d-weights}, one can compute in \LOGSPACE polynomially bounded, skew-symmetric non-zero circulation weights $w'$ for $G'$.

We construct a weight function $w$ for $\bidirect{G}$ as follows. For that, we associate with each edge $(u,v) \in \bidirect{E}$ a sequence $P(u,v)$ of edges in $\bidirect{G'}$. Recall that for each edge $(u, v)$ there is a highest bag in which both $u$ and $v$ appear. The bag above that bag contains either (a) none of the two vertices, or (b) $v$ but not $u$, or (c) $u$ but not $v$. The definition of $P(u, v)$ distinguishes these three cases:
\begin{enumerate}
 \item Suppose $B(u) = B(v)$. We set $P(u,v) = (u_{B(u)},v_{B(v)})$.
 \item Suppose $B(u)$ is a proper descendant of $B(v)$. Let $B = B(u)$ and $B' = B(v)$. 
 We set $P(u,v) = (u_{B},v_{B}), (v_B,v_{\tparent(B)}), \ldots, (v_{\tparent(\cdots(\tparent(B)))},v_{B'})$.
 \item Suppose $B(v)$ is a proper descendant of $B(u)$. Let $B = B(u)$ and $B' = B(v)$. 
 We set $P(u,v) = (u_B,u_{\tparent(\cdots(\tparent(B')))}), \ldots, (u_{\tparent(B')},u_{B'}), (u_{B'},v_{B'})$.
\end{enumerate} 
Now, let $w(u,v)$ be the sum $\sum_{e \in P(u,v)} w'(e)$ of the weights of the edges $e$ in $P(u,v)$.
Because $w'$ is a polynomially bounded skew-symmetric weight assignment, so is $w$. 

It remains to show that $w$ has non-zero circulation.
Let $\calC$ be an arbitrary simple cycle in $\bidirect{G}$. We need to show that $w$ assigns a non-zero weight to $\calC$.
Let $e_1, \ldots, e_m$ be the sequence of edges that constitutes $\calC$, and let $\calW'$ be the sequence $P(e_1), \ldots, P(e_m)$ of edges. 
By definition, the weight of $\calC$ under $w$ is the same as the weight of $\calW'$ under $w'$.

Note that $\calW'$ constitutes a cycle in $\bidirect{G'}$ which is not necessarily simple: some nodes might be visited more than once.
We show that we can construct from $\calW'$ a simple cycle $\calC'$ by removing parts of $\calW'$ with total weight $0$.
As a result, $\calC'$ has the same weight as $\calW'$ under $w'$. Because $w'$ has non-zero circulation, the weight of $\calC'$ and $\calW'$ is non-zero, and so is the weight of $\calC$ under $w$.

Suppose that some node $u_B$ is visited twice by $\calW'$. Then $\calW'$ has the subsequence $P(v,u)P(u,v')$ for some nodes $v$ and $v'$, because $u$ is visited only once in $\calC$ and no node appears twice in a single sequence $P(u,u')$.
Moreover, it most be that $h(u)$ is greater than both $h(v)$ and $h(v')$, or smaller than both $h(v)$ and $h(v')$, and either $B(v)$ is a descendent of $B(v')$ or $B(v')$ is a descendant of $B(v)$. 
We consider the case that $h(u)$ is greater than both $h(v)$ and $h(v')$, and $B(v')$ is a descendant of $B(v)$. The other cases are analogous.
Then $P(v,u)P(u,v')$ visits the nodes $v_{B(v)},u_{B(v)}, u_{\tparent(B(v))}, \ldots, u_{B(u)}, \ldots u_{\tparent(B(v))}, u_{B(v)}, \ldots, u_{B(v')}, v'_{B(v')}$ in that order. 
The closed walk from $u_{B(v)}$ to $u_{B(u)}$ and back to $u_{B(v)}$ has, because of skew-symmetry, a total weight of $0$ under $w'$. So, the corresponding edges can be removed from $\calW'$ without changing the weight. Repeating this step results in a simple cycle $\calC'$ with the same weight under $w'$ as $\calC$ under $w$. 
As $w'$ has non-zero circulation, the weight of $\calC'$ is non-zero, and so is the weight of $\calC$.
\end{proof}

\section{Conclusion}\label{section:conclusion}
The complexity of maintaining (variants of) the reachability query is the dominant research question in dynamic complexity theory.
With this paper we basically settle this question for reachability in undirected graphs, at least with respect to the size of a change: reachability in undirected graphs is in \DynFOar if and only if the changes have at most polylogarithmic size.
For reachability in directed graphs, we can only show this for insertions of polylogarithmic size, and the main open problem is whether this can be extended to also allow for deletions of single edges, non-constantly many edges, or even polylogarithmically many edges. 

We give preliminary results for classes of graphs for which non-zero circulation weights can be computed in \AC: reachability for these graphs is in \DynFOParityar under insertions and deletions of polylogarithmic size.
We show that one can compute such weight assignments for graphs with bounded treewidth.
Other graph classes for which this is possible include the class of planar graphs \cite{TewariV12}, and in general all graphs with bounded genus, which one can show using results from \cite{DattaKTV12}.

A question for further research is whether reachability for classes of directed graphs can be maintained in \DynFOar under insertions and deletions of polylogarithmic size. Candidate classes are graphs with bounded treewidth, and directed acyclic graphs.

\bibliography{bibliography}

\begin{thebibliography}{10}

\bibitem{BarringtonIS90}
David A.~Mix Barrington, Neil Immerman, and Howard Straubing.
\newblock On uniformity within {NC}{\({^1}\)}.
\newblock {\em J. Comput. Syst. Sci.}, 41(3):274--306, 1990.
\newblock \href {http://dx.doi.org/10.1016/0022-0000(90)90022-D}
  {\path{doi:10.1016/0022-0000(90)90022-D}}.

\bibitem{ChenOST16}
Xi~Chen, Igor~Carboni Oliveira, Rocco~A. Servedio, and Li{-}Yang Tan.
\newblock Near-optimal small-depth lower bounds for small distance
  connectivity.
\newblock In Daniel Wichs and Yishay Mansour, editors, {\em Proceedings of the
  48th Annual {ACM} {SIGACT} Symposium on Theory of Computing, {STOC} 2016,
  Cambridge, MA, USA, June 18-21, 2016}, pages 612--625. {ACM}, 2016.
\newblock \href {http://dx.doi.org/10.1145/2897518.2897534}
  {\path{doi:10.1145/2897518.2897534}}.

\bibitem{DattaHK14}
Samir Datta, William Hesse, and Raghav Kulkarni.
\newblock Dynamic complexity of directed reachability and other problems.
\newblock In Javier Esparza, Pierre Fraigniaud, Thore Husfeldt, and Elias
  Koutsoupias, editors, {\em Automata, Languages, and Programming - 41st
  International Colloquium, {ICALP} 2014, Copenhagen, Denmark, July 8-11, 2014,
  Proceedings, Part {I}}, volume 8572 of {\em Lecture Notes in Computer
  Science}, pages 356--367. Springer, 2014.
\newblock \href {http://dx.doi.org/10.1007/978-3-662-43948-7\_30}
  {\path{doi:10.1007/978-3-662-43948-7\_30}}.

\bibitem{DattaKMSZ18}
Samir Datta, Raghav Kulkarni, Anish Mukherjee, Thomas Schwentick, and Thomas
  Zeume.
\newblock Reachability is in {DynFO}.
\newblock {\em J. ACM}, 65(5):33:1--33:24, August 2018.
\newblock \href {http://dx.doi.org/10.1145/3212685}
  {\path{doi:10.1145/3212685}}.

\bibitem{DattaKTV12}
Samir Datta, Raghav Kulkarni, Raghunath Tewari, and N.~V. Vinodchandran.
\newblock Space complexity of perfect matching in bounded genus bipartite
  graphs.
\newblock {\em J. Comput. Syst. Sci.}, 78(3):765--779, 2012.
\newblock \href {http://dx.doi.org/10.1016/j.jcss.2011.11.002}
  {\path{doi:10.1016/j.jcss.2011.11.002}}.

\bibitem{DattaMSVZ19}
Samir Datta, Anish Mukherjee, Thomas Schwentick, Nils Vortmeier, and Thomas
  Zeume.
\newblock {A Strategy for Dynamic Programs: Start over and Muddle through}.
\newblock {\em {Logical Methods in Computer Science}}, {Volume 15, Issue 2},
  May 2019.
\newblock \href {http://dx.doi.org/10.23638/LMCS-15(2:12)2019}
  {\path{doi:10.23638/LMCS-15(2:12)2019}}.

\bibitem{DattaMVZ18}
Samir Datta, Anish Mukherjee, Nils Vortmeier, and Thomas Zeume.
\newblock Reachability and distances under multiple changes.
\newblock In Ioannis Chatzigiannakis, Christos Kaklamanis, D{\'{a}}niel Marx,
  and Donald Sannella, editors, {\em 45th International Colloquium on Automata,
  Languages, and Programming, {ICALP} 2018, July 9-13, 2018, Prague, Czech
  Republic}, volume 107 of {\em LIPIcs}, pages 120:1--120:14. Schloss Dagstuhl
  - Leibniz-Zentrum fuer Informatik, 2018.
\newblock \href {http://dx.doi.org/10.4230/LIPIcs.ICALP.2018.120}
  {\path{doi:10.4230/LIPIcs.ICALP.2018.120}}.

\bibitem{DattaMVZ18a}
Samir Datta, Anish Mukherjee, Nils Vortmeier, and Thomas Zeume.
\newblock Reachability and distances under multiple changes.
\newblock {\em CoRR}, abs/1804.08555, 2018.
\newblock \href {http://arxiv.org/abs/1804.08555} {\path{arXiv:1804.08555}}.

\bibitem{DongS98}
Guozhu Dong and Jianwen Su.
\newblock Arity bounds in first-order incremental evaluation and definition of
  polynomial time database queries.
\newblock {\em J. Comput. Syst. Sci.}, 57(3):289--308, 1998.
\newblock \href {http://dx.doi.org/10.1006/jcss.1998.1565}
  {\path{doi:10.1006/jcss.1998.1565}}.

\bibitem{ElberfeldJT10}
Michael Elberfeld, Andreas Jakoby, and Till Tantau.
\newblock Logspace versions of the theorems of {B}odlaender and {C}ourcelle.
\newblock In {\em 51th Annual {IEEE} Symposium on Foundations of Computer
  Science, {FOCS} 2010, October 23-26, 2010, Las Vegas, Nevada, {USA}}, pages
  143--152. {IEEE} Computer Society, 2010.
\newblock \href {http://dx.doi.org/10.1109/FOCS.2010.21}
  {\path{doi:10.1109/FOCS.2010.21}}.

\bibitem{FKS}
Michael~L. Fredman, J{\'{a}}nos Koml{\'{o}}s, and Endre Szemer{\'{e}}di.
\newblock Storing a sparse table with {O(1)} worst case access time.
\newblock {\em J. {ACM}}, 31(3):538--544, 1984.
\newblock \href {http://dx.doi.org/10.1145/828.1884}
  {\path{doi:10.1145/828.1884}}.

\bibitem{HealyV06}
Alexander Healy and Emanuele Viola.
\newblock Constant-depth circuits for arithmetic in finite fields of
  characteristic two.
\newblock In Bruno Durand and Wolfgang Thomas, editors, {\em {STACS} 2006, 23rd
  Annual Symposium on Theoretical Aspects of Computer Science, Marseille,
  France, February 23-25, 2006, Proceedings}, volume 3884 of {\em Lecture Notes
  in Computer Science}, pages 672--683. Springer, 2006.
\newblock \href {http://dx.doi.org/10.1007/11672142\_55}
  {\path{doi:10.1007/11672142\_55}}.

\bibitem{HendersonS81}
Harold~V Henderson and Shayle~R Searle.
\newblock On deriving the inverse of a sum of matrices.
\newblock {\em Siam Review}, 23(1):53--60, 1981.

\bibitem{Hesse03Reach}
William Hesse.
\newblock The dynamic complexity of transitive closure is in
  {DynTC}\({}^{\mbox{0}}\).
\newblock {\em Theor. Comput. Sci.}, 296(3):473--485, 2003.
\newblock \href {http://dx.doi.org/10.1016/S0304-3975(02)00740-5}
  {\path{doi:10.1016/S0304-3975(02)00740-5}}.

\bibitem{HesseAB02}
William Hesse, Eric Allender, and David A.~Mix Barrington.
\newblock Uniform constant-depth threshold circuits for division and iterated
  multiplication.
\newblock {\em J. Comput. Syst. Sci.}, 65(4):695--716, 2002.
\newblock \href {http://dx.doi.org/10.1016/S0022-0000(02)00025-9}
  {\path{doi:10.1016/S0022-0000(02)00025-9}}.

\bibitem{Immerman88}
Neil Immerman.
\newblock Nondeterministic space is closed under complementation.
\newblock {\em {SIAM} J. Comput.}, 17(5):935--938, 1988.
\newblock \href {http://dx.doi.org/10.1137/0217058}
  {\path{doi:10.1137/0217058}}.

\bibitem{ImmermanDC}
Neil Immerman.
\newblock {\em Descriptive complexity}.
\newblock Graduate texts in computer science. Springer, 1999.
\newblock \href {http://dx.doi.org/10.1007/978-1-4612-0539-5}
  {\path{doi:10.1007/978-1-4612-0539-5}}.

\bibitem{Jukna2012}
Stasys Jukna.
\newblock {\em Boolean function complexity: advances and frontiers}, volume~27.
\newblock Springer Science \& Business Media, 2012.

\bibitem{KallampallyT16}
Vivek Anand~T. Kallampally and Raghunath Tewari.
\newblock Trading determinism for time in space bounded computations.
\newblock In {\em 41st International Symposium on Mathematical Foundations of
  Computer Science, {MFCS} 2016, August 22-26, 2016 - Krak{\'{o}}w, Poland},
  pages 10:1--10:13, 2016.
\newblock \href {http://dx.doi.org/10.4230/LIPIcs.MFCS.2016.10}
  {\path{doi:10.4230/LIPIcs.MFCS.2016.10}}.

\bibitem{PatnaikI97}
Sushant Patnaik and Neil Immerman.
\newblock {Dyn-FO}: A parallel, dynamic complexity class.
\newblock {\em J. Comput. Syst. Sci.}, 55(2):199--209, 1997.
\newblock \href {http://dx.doi.org/10.1006/jcss.1997.1520}
  {\path{doi:10.1006/jcss.1997.1520}}.

\bibitem{PTV}
Aduri Pavan, Raghunath Tewari, and N.~V. Vinodchandran.
\newblock On the power of unambiguity in log-space.
\newblock {\em Computational Complexity}, 21(4):643--670, 2012.
\newblock \href {http://dx.doi.org/10.1007/s00037-012-0047-3}
  {\path{doi:10.1007/s00037-012-0047-3}}.

\bibitem{SchwentickZ16}
Thomas Schwentick and Thomas Zeume.
\newblock Dynamic complexity: recent updates.
\newblock {\em {SIGLOG} News}, 3(2):30--52, 2016.
\newblock URL: \url{http://doi.acm.org/10.1145/2948896.2948899}.

\bibitem{Smolensky1987}
Roman Smolensky.
\newblock Algebraic methods in the theory of lower bounds for boolean circuit
  complexity.
\newblock In {\em Proceedings of the nineteenth annual ACM symposium on Theory
  of computing}, pages 77--82, 1987.
\newblock \href {http://dx.doi.org/10.1145/28395.28404}
  {\path{doi:10.1145/28395.28404}}.

\bibitem{Szelep1987}
R{\'{o}}bert Szelepcs{\'{e}}nyi.
\newblock The method of forced enumeration for nondeterministic automata.
\newblock {\em Acta Inf.}, 26(3):279--284, 1988.
\newblock \href {http://dx.doi.org/10.1007/BF00299636}
  {\path{doi:10.1007/BF00299636}}.

\bibitem{TewariV12}
Raghunath Tewari and N.~V. Vinodchandran.
\newblock Green's theorem and isolation in planar graphs.
\newblock {\em Inf. Comput.}, 215:1--7, 2012.
\newblock \href {http://dx.doi.org/10.1016/j.ic.2012.03.002}
  {\path{doi:10.1016/j.ic.2012.03.002}}.

\bibitem{Vortmeier19}
Nils Vortmeier.
\newblock {\em Dynamic expressibility under complex changes}.
\newblock PhD thesis, TU Dortmund University, Germany, 2019.
\newblock \href {http://dx.doi.org/10.17877/DE290R-20434}
  {\path{doi:10.17877/DE290R-20434}}.

\end{thebibliography}

\end{document}